\documentclass[a4paper,twoside]{article}

\usepackage{epsfig}
\usepackage{subcaption}
\usepackage{calc}
\usepackage{amssymb}
\usepackage{amstext}
\usepackage{amsmath}
\usepackage{amsthm}

\usepackage{multicol}
\usepackage{pslatex}

\usepackage{graphics} 

\usepackage{hyperref}
\usepackage{cleveref}

\usepackage{apalike}
\usepackage{algorithm2e}
\usepackage[bottom]{footmisc}

\usepackage{booktabs}
\usepackage[obeyspaces]{xurl}
\usepackage{comment}
\usepackage[export]{adjustbox}
\usepackage{tabulary}
\usepackage{tabularx}

\newtheorem{prop}{Proposition}

\newtheorem{defi}{Definition}

\usepackage{mathtools}

\newcommand{\Real}{\mathbb{R}}

\newcommand{\cS}{\mathcal{S}}
\newcommand{\cA}{\mathcal{A}}

\newcommand{\cC}{\mathcal{C}}

\newcommand{\cN}{\mathcal{N}}
\newcommand{\cE}{\mathcal{E}}
\newcommand{\cR}{\mathcal{R}}
\newcommand{\cT}{\mathcal{T}}
\newcommand{\cU}{\mathcal{U}}
\newcommand{\cX}{\mathcal{X}}
\newcommand{\cY}{\mathcal{Y}}

\newcommand{\norm}[1]{\lVert #1 \rVert}

\usepackage[dvipsnames]{xcolor}

\usepackage{SCITEPRESS}     % Please add other packages that you may need BEFORE the SCITEPRESS.sty package.

\usepackage[a4paper,margin=1in]{geometry}
\usepackage{fancyhdr}
\begin{document}

\title{C-STEP: Continuous Space-Time Empowerment for 
Physics-informed Safe Reinforcement Learning of Mobile Agents}

\author{\authorname{Guihlerme Daubt\sup{1}\orcidAuthor{0009-0002-3208-5734} and Adrian Redder\sup{2}\orcidAuthor{0000-0001-7391-4688}}
\affiliation{\sup{1}Faculty of Technology, Bielefeld University, Germany}
\affiliation{\sup{2}Department of Electrical Engineering and Information Technology, Paderborn University, Germany}
\email{ gnmdaudt@gmail.com , aredder@mail.upb.de}
}

\keywords{Reinforcement Learning, Physics-Informed Learning, Empowerment, Reward Shaping, Safe Navigation, Mobile Robotic Systems, Mobile Agents}

\abstract{
Safe navigation in complex environments remains a central challenge for reinforcement learning (RL) in robotics. This paper introduces Continuous Space-Time Empowerment for Physics-informed (C-STEP) safe RL, a novel measure of agent-centric safety tailored to deterministic, continuous domains. This measure can be used to design physics-informed intrinsic rewards by augmenting positive navigation reward functions. The reward incorporates the agent’s internal states (e.g., initial velocity) and forward dynamics to differentiate safe from risky behavior. By integrating C-STEP with navigation rewards, we obtain an intrinsic reward function that jointly optimizes task completion and collision avoidance. Numerical results demonstrate fewer collisions, reduced proximity to obstacles, and only marginal increases in travel time. Overall, C-STEP offers an interpretable, physics-informed approach to reward shaping in RL, contributing to safety for agentic mobile robotic systems.
}

\onecolumn \maketitle \normalsize \setcounter{footnote}{0} \vfill
\thispagestyle{fancy}
\section{\uppercase{Introduction}}
\label{sec:introduction}

Agentic mobile robotic systems require safe and robust control in dynamic environments to ensure reliable navigation. Advances in reinforcement learning (RL) have shown promise for developing such highly adaptive policies. Previous works have explored safe RL methods by using control barrier functions \cite{cbf_saferl,redder2025vds}, robustness prioritization \cite{queeney2024optimal}, and risk-avoidance \cite{Chow2015RiskConstrainedRL}. Nevertheless, safety remains a challenge when dealing with continuous space-time dynamics and unknown environmental conditions. Recent surveys have highlighted the potential of concepts from information theory for safe RL \cite{saferl_gu,saferl_learntosafe}.

Considering this, we introduce a novel formulation of empowerment \cite{klyubin2005all} for dynamical systems operating in continuous space and time. Empowerment is a measure of intrinsic agent control capacity in an environment and has traditionally been used in discrete-time stochastic computer networks. Herein, we propose a new definition of empowerment tailored for dynamical systems. By approximating empowerment through trajectory sampling and forward system dynamics, our approach captures the complex interaction between control strategies and safety, yielding an interpretable safety metric. With the empowerment approximation, we then formulate an intrinsic safe RL reward design.

Our core contribution is a physics-informed intrinsic reward design that incorporates internal system states (such as velocity) to assess safety in navigation via the local empowerment approximation of a system. This empowerment approach can be integrated with positive navigation reward functions, serving as a complementary enhancement to promote physics-informed safety rather than replacing conventional reward designs. The rest of this paper is structured as follows. \Cref{sec:background} provides background material. \Cref{sec:related_work} discusses prior research. \Cref{sec:methodology} introduces our methodology, starting with a new definition of empowerment for deterministic systems and culminating in our proposed C-STEP framework. \Cref{sec:numerical} presents numerical experiments highlighting the safety benefits of our approach. Finally, we discuss C-STEP in comparison to other safe RL paradigms in \Cref{sec:discussion} and conclude with remarks on future work in \Cref{sec:concl}.

\section{Background}\label{sec:background}

We now recall background on information theory and RL. For in-depth background, see \cite{cover2006elements} and \cite{Sutton1998}, respectively.

\subsection{Information Theory}\label{sec:infotheory_background}

The \textit{differential entropy} of a continuous random variable $Y$ with probability density function $f(y)$ and support set $\cY \subset \Real^n$ is defined as:
\begin{equation}
    h(Y) \coloneqq -\int_{\cY} f(y) \log \left(f(y)\right) \mathrm{d}y.
\end{equation}
An important property is that differential entropy is maximized by a uniform distribution \cite[Chap. 12]{cover2006elements}, i.e., for any $Y$ with compact support $\cY$:
\begin{equation}
    \label{eq:diff_entropy_max}
    \max\limits_{p(y)} h(Y) = \log(\lambda(\cY)),
\end{equation}
where $\lambda(\cdot)$ is the $n$-dimensional volume (Lebesgue measure).

 The \textit{conditional differential entropy} $h(Y\mid X)$ of random variables $X,Y$ with probability density function $g(x)$ and conditional probability density $h(y\mid x)$is defined as:
\begin{equation}
    \label{eq:cond_entropy}
    h(Y \mid X)\coloneqq-\int_\cX g(x) \int_\cY h(y \mid x) \log (h(y \mid x)) \mathrm{d} y \mathrm{~d}x.
\end{equation}
The \textit{mutual information} is then defined as $I(Y; X) \coloneqq h(Y)-h(Y\mid X)$. Finally, the \textit{channel capacity} from an ``input'' random variable $X$ to an ``output'' variable $Y$ is the maximal possible mutual information:
\begin{equation}
    \cC(X \to Y) \coloneqq \max\limits_{p(x)} I(Y; X).
\end{equation} 
In the following sections, we adapt these core concepts of channel capacity and entropy to formulate a novel, agent-centric safety measure for deterministic robotic systems. Our central observation is that the channel capacity $\cC(X \to Y)$ can quantify safety. A system state is considered safe if the channel capacity from system control inputs to successive system states is sufficiently large. For deterministic systems, this capacity will be related to the logarithmic volume of reachable sets over a finite horizon.

\subsection{Reinforcement Learning}
\label{subsec:RL}
In RL, an agent learns to act in an environment defined by a Markov Decision Process (MDP) with state space $\cS$, action space $\cA$, transition kernel $p$, and reward function $r:\cS \times \cA \to \Real$. At each state $s \in \cS$, the agent picks an action $a\in \cA$, receives a reward $r(s,a)$, and transitions to a successor state $s'$ with distribution $p(\cdot\mid s,a)$. The objective is to find a policy $\pi: \mathcal{S} \rightarrow \mathcal{A}$ that maximizes the discounted return $R_0 \coloneqq \sum_{k=0}^{\infty} \gamma^{k}r(s_k,a_k)$. In model-free Deep RL, a policy $\pi_{\theta}$ is parameterized by a deep neural network and improved iteratively using interaction data $(s,a,r,s')$ to approximate a policy gradient.

\section{Related Work}\label{sec:related_work}

Several studies have incorporated information-theoretic measures to promote exploration and skill acquisition in RL. In \cite{intrinsically_empowerment_paper}, a variational lower bound of empowerment is used as an intrinsic reward for self-motivated learning. Building on this, \cite{hierarchical_emp_paper} introduced a hierarchical empowerment framework for skill learning. Other works have applied empowerment to enhance self-optimizing networks \cite{empowerment_networks} and to incentivize exploration in sparse-reward scenarios \cite{empowerment_drive_paper,environment_intrinsic_paper}. These methods typically approximate empowerment for discrete-time agents, often at high computational cost. Crucially, our work appears to be the first to leverage readily available forward dynamics of continuous-time mobile systems to approximate empowerment.

Another research direction integrates physics-based insights into reward design. For instance, \cite{physics_intrinsic_reward} used principles like energy conservation as intrinsic rewards. Other methods use reward functions incorporating model dynamics to facilitate diverse motion learning in complex agents \cite{physics_control_reward_function} or evaluate performance based on velocity relative to system dynamics and a target \cite{LyapunovDRL}. In contrast, our method evaluates safety by leveraging the agent's internal states and sampled forward dynamics to approximate reachable sets. We then introduce a multiplicative empowerment term applicable to any positive reward function. This method can be applied on top of reward designs and in conjunction with various safe RL methods \cite{gu2024review}. This has not been explored in prior work, which used additive penalties to promote safety.

\section{C-STEP: Continuous Space-Time Empowerment}
\label{sec:methodology}

This section presents our core contribution, the Continuous Space-Time Empowerment for Physics-informed (C-STEP) safe RL framework. We begin by illustrating why the classical definition of empowerment is ill-suited for deterministic systems and propose a re-formulation. We then develop an intrinsic, physics-informed reward for safe navigation and provide a practical sampling-based method for its approximation.

\subsection{Rethinking Empowerment for Deterministic Systems}
\label{subsec:rethinking_empowerment}
Empowerment, first introduced in \cite{klyubin2005all}, measures the control an agent has over its future sensor states. For a discrete-time perception-action loop (PAL), where at every time $k$ an agent picks an action $U_k$ based on sensor state $X_k$, leading to a subsequent state $X_{k+1}$, empowerment is defined as the conditional channel capacity from actions to subsequent states.
\begin{defi}[State-Dependent Empowerment \cite{salge2014empowerment}]{}
    \label{def:emp_discrete_time}
    For a discrete-time PAL, state-dependent empowerment is:
    \begin{equation}
        \cE(x) \coloneqq \cC\left(U_k \to X_{k+1} \mid X_k = X\right).
\end{equation}
\end{defi}
A problem with \Cref{def:emp_discrete_time} arises when state transitions are deterministic in continuous spaces. In this setting, the relationship between an action and the next state is perfectly predictable, causing the conditional entropy $h(X_{k+1} \mid U_k)$ to become $-\infty$. Consequently, empowerment becomes infinite for any state, as a noiseless channel can encode unbounded information \cite[Section 9.3]{cover2006elements}. Therefore, empowerment as stated in \Cref{def:emp_discrete_time} cannot measure the control or influence a deterministic agent has locally over its environment.

To resolve this, we observed that adding small, state-independent noise to a deterministic system makes maximizing channel capacity equivalent to maximizing the differential entropy of the subsequent state distribution. This observation, for which we give further details in the appendix, motivates a revised definition of empowerment for deterministic agents.

\begin{defi}[Deterministic Agent Empowerment]{}
    \label{def:emp_discrete_time2}
    For a deterministic PAL, empowerment is defined as:
    \begin{equation}
        \cE(s) \coloneqq \max\limits_{p(u)} h(X_{k+1} \mid X_k = X).
    \end{equation}
\end{defi}
From Eq.\eqref{eq:diff_entropy_max} it follows that the maximum empowerment for deterministic agents in Definition \eqref{def:emp_discrete_time2} is achieved when actions induce a uniform distribution over the set of reachable states. Empowerment thus becomes the logarithm of the volume of this reachable set.

\subsection{From Empowerment to Safe RL Rewards}
\label{subsec:empowerment_to_reward}

We formalize our framework for a non-autonomous continuous-time dynamical system:
\begin{equation}
\label{eq:ODE}
    \dot{x}(t) = f(x(t),u(t)),
\end{equation}
with states $x \in \Real^n$, controls $u \in \cU \subset \Real^m$, and a collision-free configuration space $\cX_{\text{free}} \subset \Real^n$, such that any $x \in \cX_{\text{free}}$ is a feasible state. Then, for any initial condition of the ODE in $\cX_{\text{free}}$, we define CST-Empowerment as the deterministic agent empowerment where the action space is the space of all admissible continuous finite-horizon input trajectories. Key definitions and notation for this section are summarized in \Cref{def:key_defs}.
\begin{table}[!t]
\centering
\captionsetup{font=footnotesize}
\captionof{table}{Key Definitions and Notation}
\begin{tabularx}{0.95\columnwidth}{lX}
\toprule
\textbf{Symbol} & \textbf{Definition} \\
\midrule
$\cX_{\text{free}}$ & The collision-free configuration space. \\
$\cR_T(x)$ & Set of reachable states in $\Real^n$ from state $x$ within time $T$ without state constraints. \\
$\cR_{T,\text{free}}(x)$ & Set of reachable states where the state trajectory remains in $\cX_{\text{free}}$. \\
$\cT_T(x)$ & The terminal set, $\cR_T(x) \setminus \cR_{T,\text{free}}(x)$. \\
%$S(\cdot)$ & A sensor function: $1$ if a trajectory is collision-free, $0$ otherwise. \\
$c$ & A positive safety coefficient scaling the empowerment reward. \\
\bottomrule
\end{tabularx}
\label{def:key_defs}
\end{table}

\begin{defi}[CST-Empowerment]
\label{def:T-Emp}
For any $T>0$, the $T$-Horizon CST-Empowerment of \eqref{eq:ODE} is defined as
    \begin{equation}
        \cE_T(x) \coloneqq \max\limits_{p(u(\cdot) )} h(X(T)) \mid x(0)=x_0 ), \quad \forall x_0 \in \cX_{\text{free}},
        \nonumber
    \end{equation}
    where $X(T)$ is the random variable that describes the final state of a $T$-horizon solution trajectory $x(t)$ that starts from some $x_0 \in \cX_{\text{free}}$ and remains in $\cX_{\text{free}}$ given input trajectories from some distribution $p(u(\cdot))$.
\end{defi}
Given that differential entropy is maximized for uniform distributions, we can now conclude that the CST-Empowerment of any state of a first-order non-autonomous ODE equals the log-volume of the $T$-reachable set of trajectories starting from that state and remaining in $\cX_{\text{free}}$. 
\begin{prop}
    \label{prop:1}
    Suppose that $f$ is Lipschitz continuous, then the CST-Empowerment of \eqref{eq:ODE} is given by
    \begin{align}
    \label{eq:emp_sys1}
        \cE_T(x) &= \log\left(\lambda(\cR_{T,\text{free}}(x))\right),  \\
        &=\log\left(\lambda(\cR_T(x))-\lambda(\cT_T(x))\right),
        \nonumber
    \end{align}
    for all $x\in \Real^n$ with $\lambda(\cdot)$ the $n$-dim. volume on $\Real^n$
    %and the reachable set
    %\begin{equation}
    %    \cR_T\left(x, \mathcal{U}\right)\coloneqq\left\{x' \in \Real^n \mid \exists u(\cdot) : x(0)=x, x(T)=x'\right\},
    %    \nonumber
    %\end{equation}
    %where $u(\cdot)$ is an admissible continuous $T$-Horizon control trajectory. 
\end{prop}
\begin{proof}
    See Appendix.
\end{proof}

\begin{algorithm}[!t]
 \caption{C-STEP: Sampling-based CST-Empowerment approximation algorithm}
 \label{algo:emp_approx}
 \textbf{Input:} Current state $x$, time horizon $T$, number of samples $N$\;
 Sample $N$ admissible control trajectories $u(\cdot): [0,T] \to \cU$\;
 Calculate solution trajectories $x(\cdot)$ for each $u(\cdot)$ with $x(0)=x$.\;
 Initialize empty sets $\cR \leftarrow \emptyset$, $\cT \leftarrow \emptyset$\;
 \For{each solution trajectory $x(\cdot)$}{
  \eIf{trajectory is collision-free %(i.e., $S(x(\cdot) - x) = 1$}{
   }{Add final position $x_{pos}(T)$ to $\cR$\;
  }{Add final position $x_{pos}(T)$ to $\cT$\;}
 }
 Compute estimate $\hat{\cR}_T(x)$ of $\cR_T(x)$ from $\cR$\;
 Compute estimate $\hat{\cT}_T(x)$ of $\cT_T(x)$ from $\cT$\;
 \textbf{Return} C-STEP Approximation \\ $\hat{\cE}_T(x) = \log\left(\lambda(\hat{\cR}_T(x)) - \lambda(\hat{\cT}_T(x))\right)$\;
\end{algorithm}      

Based on \Cref{prop:1}, we propose Algorithm \ref{algo:emp_approx} to obtain practical, efficient approximations of CST-Empowerment using sampled trajectories. Further details on the sampling-based approximation are discussed in \Cref{subsec:sampling}.

We now leverage the availability of CST-Empowerment approximations based on Algorithm \ref{algo:emp_approx} to create an intrinsic reward for safe control of continuous time ODEs \eqref{eq:ODE}. By multiplying a task-specific reward with our safety measure, i.e., CST-Empowerment, we encourage joint optimization of task completion and safety. To give a user a tool to tune the degree of safety, we introduce a safety coefficient $c>0$ that scales the empowerment. 

\begin{defi}[Empowered Navigation Reward Function]{}
\label{def:emp_reward}
Given a task reward $r_d:\Real^n \to [0,1]$, the empowered reward is defined as:
    \begin{equation}
    \label{eq:empowered_reward}
        r(x) \coloneqq r_d(x)\log\left(c \lambda(\cR_{T,\text{free}}(x))\right).
    \end{equation}
    with safety coefficient $c>0$.
\end{defi}
Our empowered reward design, \Cref{def:emp_reward}, has several key properties that we summarize next.

\textbf{Technical properties:} CST-Empowerment weights the navigation reward by the $\log$ of the reachable volume. Thus, empowered rewards characterize a state's safety by the reachable volume from that state. Decreasing the coefficient $c > 0$ increases the safety of a learned policy; see below for a more specific connection between $c$ and safety.

\textbf{Physics-informed Intrinsic Design:} CST-Empowerment incorporates the forward dynamics of a system into a measure of how much space the system can reach from a starting state. It is intrinsic as it is specific to the given system dynamics.

\textbf{Dependence on Internal States:} Most reward designs for navigation incorporate external quantities—such as position or velocity—to penalize unsafe conditions. However, these quantities alone may be misleading as internal states may be close to becoming unstable. By contrast, our empowered reward design captures the influence of all internal states by leveraging the full initial state and forward dynamics in the reachability calculation.

\textbf{Environment Independence:} To approximate CST-Empowerment, it is only necessary to have a sensor function or an estimation to determine/approximate whether a forward trajectory of a system leads to termination. No global map or prior information about obstacles is required.

\textbf{Interpretability:} If an RL agent is trained solely to maximize CST-Empowerment $r(x) = \cE_T(x)$, its optimal behavior is to navigate to states from which it can reach the largest possible volume, promoting inherently safe and flexible positioning. More generally, suppose an agent is trained to maximize \eqref{eq:empowered_reward} and empirically a positive reward is observed in the mean after training. Then one can conclude that on average $\lambda(\cR_{T,\text{free}}(x)) > c$, which can be translated to safety margins for the dynamical system.

\subsection{Sampling-based Approximation and Hyperparameters}
\label{subsec:sampling}
Calculating the exact volume of $\cR_{T,\text{free}}(x)$ is often intractable. We therefore propose a practical sampling-based approximation: First, we compute the convex hull of reachable endpoints to approximate $\hat{\cR}_T(x)$, second, we compute the union of convex hulls of clustered terminal points to approximate $\hat{\cT}_T(x)$.
\Cref{fig:convex_subtract} illustrates how this approximation effectively captures the impact of initial velocity on safety.

Our sampling-based approximation is agnostic to the method used to sample control trajectories. While it is possible to characterize how the approximation error decreases as the number of sampled trajectories increases, this analysis lies outside the scope of the present work. Although sampling many trajectories can be appealing in principle, it may become computationally prohibitive for complex non-linear systems. Consequently, the choice of how many trajectories to sample must be tuned according to the problem and available resources.

A more critical hyperparameter than the number of trajectories is the time horizon $T$ for each trajectory. A large $T$ risks all trajectories colliding (zero empowerment), while a short $T$ offers little information. We propose a heuristic that has worked well for our numerical experiments: set $T$ to the system's worst-case stopping time from maximum velocity. This provides a reasoned baseline that can be tuned for more conservative or aggressive behavior.

\section{Numerical Experiments}
\label{sec:numerical}
We present 2D and 3D navigation tasks where agents must reach targets while avoiding collisions. We compare agents trained with a standard navigation reward ("unempowered") against one trained with our empowered reward function ("empowered"). For every experiment at every time step, we use the standard navigation reward  $\exp{\left(-\norm{\text{agent position} - \text{target position}}_2\right)}$.

Note that the unempowered agents still receive safety feedback because episodes terminate upon collision; thus, value function estimates will not be bootstrapped. It is important to note that our goal in this preliminary study is not to compare with other safe RL methods, but to understand the general behaviour of empowered agents in basic navigation tasks.
For all experiments, we use the Proximal Policy Optimization (PPO) algorithm \cite{ppo_paper} from Stable Baselines3 \cite{stable_baselines3}. The code is available on GitHub\footnote{\url{https://github.com/gdaudt/empowerment-rl}.}.

\subsection{Point Maze Environment}
Our first experiment considers a Point Maze environment where a 2-DoF ball must be navigated to a target by applying a 2D force. The map has two paths: one narrow and one wide. The goal is placed so that the narrow path is the shortest route. The behaviour of the agents after training is shown in \Cref{fig:pointmaze}. The unempowered agent learns to take the faster, riskier path, resulting in more collisions. In contrast, the empowered agent learns to take the longer but safer route, thereby demonstrating a behavioral change toward prioritizing safety.

\begin{figure}[htb]
  \centering
  \begin{subfigure}[b]{0.4\columnwidth} 
    \centering
    \includegraphics[width=\textwidth, height=5cm]{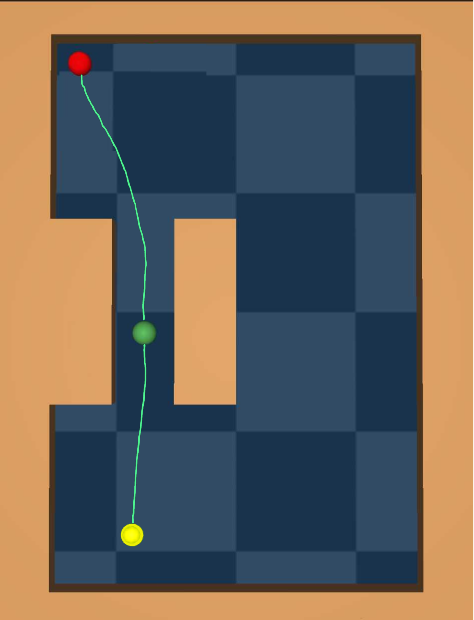} 
    \subcaption{Behavior of the unempowered agent.}
    \label{fig:pointmaze-noemp}
  \end{subfigure}
  \hspace{2mm} 
  \begin{subfigure}[b]{0.4\columnwidth} 
    \centering
    \includegraphics[width=\textwidth, height=5cm]{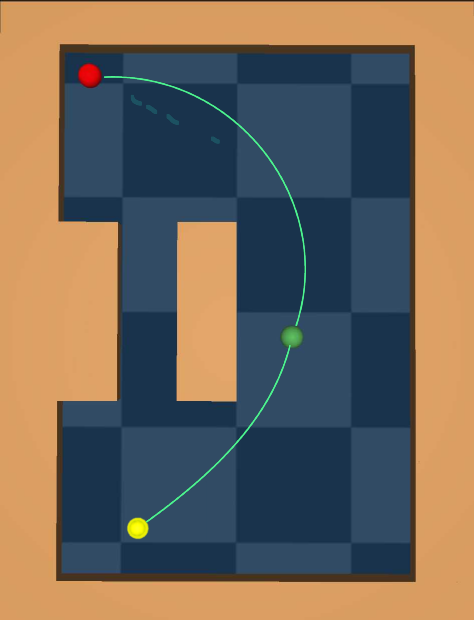} 
    \subcaption{Behavior of the empowered agent.}
    \label{fig:pointmaze-emp}
  \end{subfigure}
  \caption{Point maze navigation. Green spheres are agents, red is the goal, and yellow is the start. The empowered agent prefers the safer, wider path.}
  \label{fig:pointmaze}
\end{figure}

\begin{figure}[htb]
  \centering
  \begin{subfigure}[c]{0.98\columnwidth}
    \centering
    \includegraphics[height=0.15\textheight, valign=t, center]{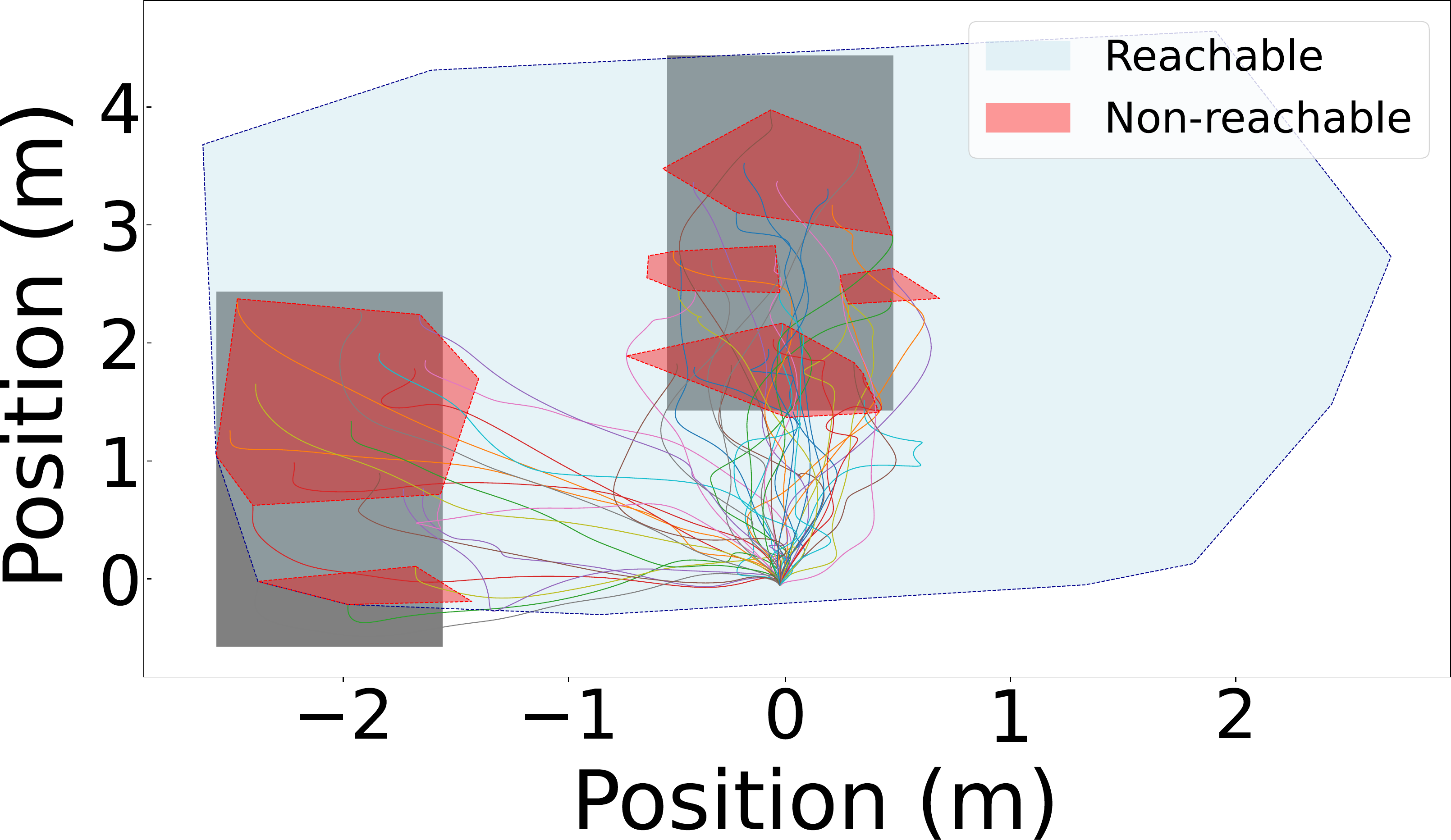}
    \captionsetup{width=1\linewidth}\caption{Initial velocity (m/s) $(v_x \ v_y \ v_z) = (0 \ 2 \ 0)$.}
    \label{fig:convex-y2}
  \end{subfigure}%
  \\
  \vspace{0.3cm}
  \begin{subfigure}[c]{0.98\columnwidth}
    \centering
    \includegraphics[height=0.15\textheight, valign=t, center]{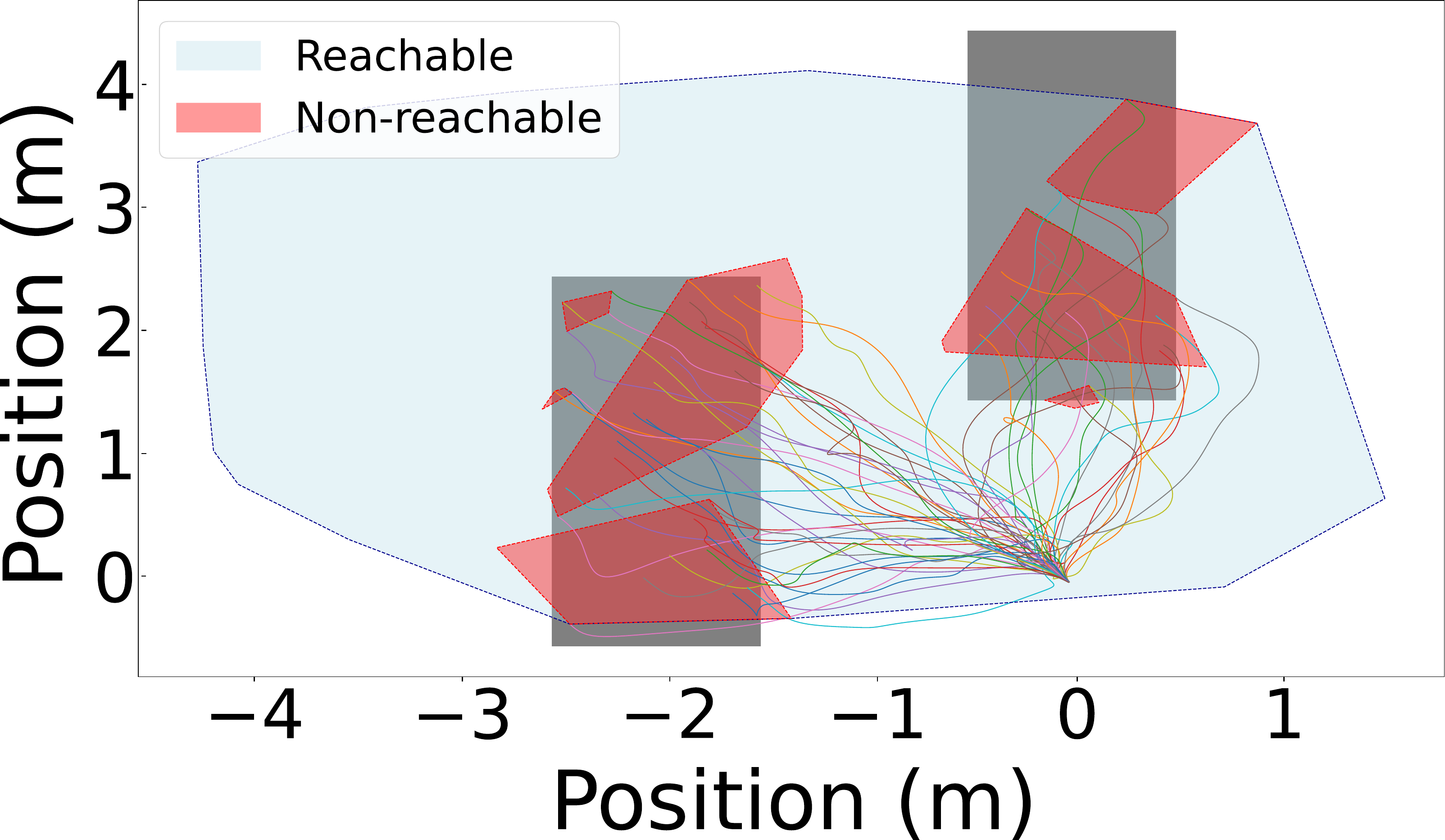}
    \captionsetup{width=1\linewidth}\caption{Initial velocity (m/s) $(v_x \ v_y \ v_z) = (1.5 \ 2 \ 0)$.}
    \label{fig:convex-x1.5y2}
  \end{subfigure}
  \caption{Visualization of the sampling-based reachable set approximation in 2D for different initial velocities. The blue areas indicate total reachable areas; red areas represent the approximated terminal set volume.}
  \label{fig:convex_subtract}
\end{figure}

\subsection{Drone PyBullet Environment}
Our second simulation considers a 3D drone in a dynamic environment (\Cref{fig:pybullet-map}). Both the empowered and the unempowered agent use relative distance to the goal, drone velocity, and 2D LiDAR data as state inputs. The map's path widths are randomized each episode, forcing reliance on sensor data. The training progress (\Cref{fig:avg-reward}) shows the empowered agent successfully learns, indicating our reward provides a meaningful learning signal.

We evaluated a unempowered agent and empowered agents with safety coefficients $c = \{0.5, 1, 10\}$ on 100 different map configurations. As shown in \Cref{tab:success_time}, all empowered agents achieved higher success rates than the unempowered baseline. The safest agent ($c=0.5$) achieved a 99\% success rate, compared to 82\% for the baseline, with only a marginal increase in the time taken to clear the obstacle.

To further quantify safety, we measured the time agents spent near obstacles (\Cref{fig:dist-threshold}). \Cref{tab:time_decrease} shows that empowered agents consistently spent less time in close proximity to walls. For instance, the agent with $c=1$ spent over 72\% less time within 0.1m of an obstacle while only being 18.4\% slower, a significant safety improvement for a small trade-off in task completion time.

\begin{figure}[htb]
    \centering
    \includegraphics[width=0.8\columnwidth]{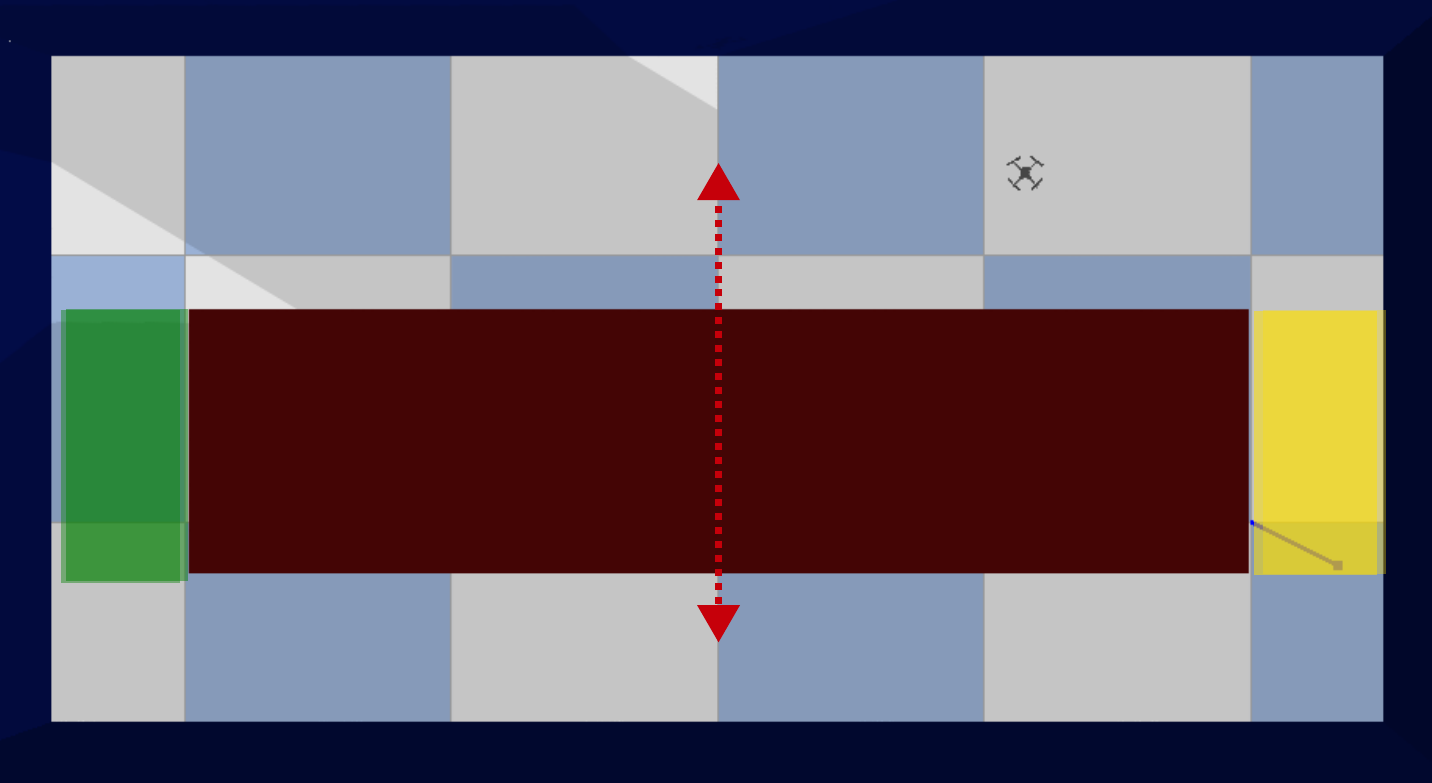}
    \caption{Top-down view of the PyBullet simulation environment. The agent starts in the yellow region and navigates to the green goal region. The maroon-colored obstacle's width and position are randomized in each episode, indicated by the red arrows, requiring adaptive navigation.}
  \label{fig:pybullet-map}
\end{figure}

\begin{figure}[htb]
    \centering
    \includegraphics[width=0.9\columnwidth]{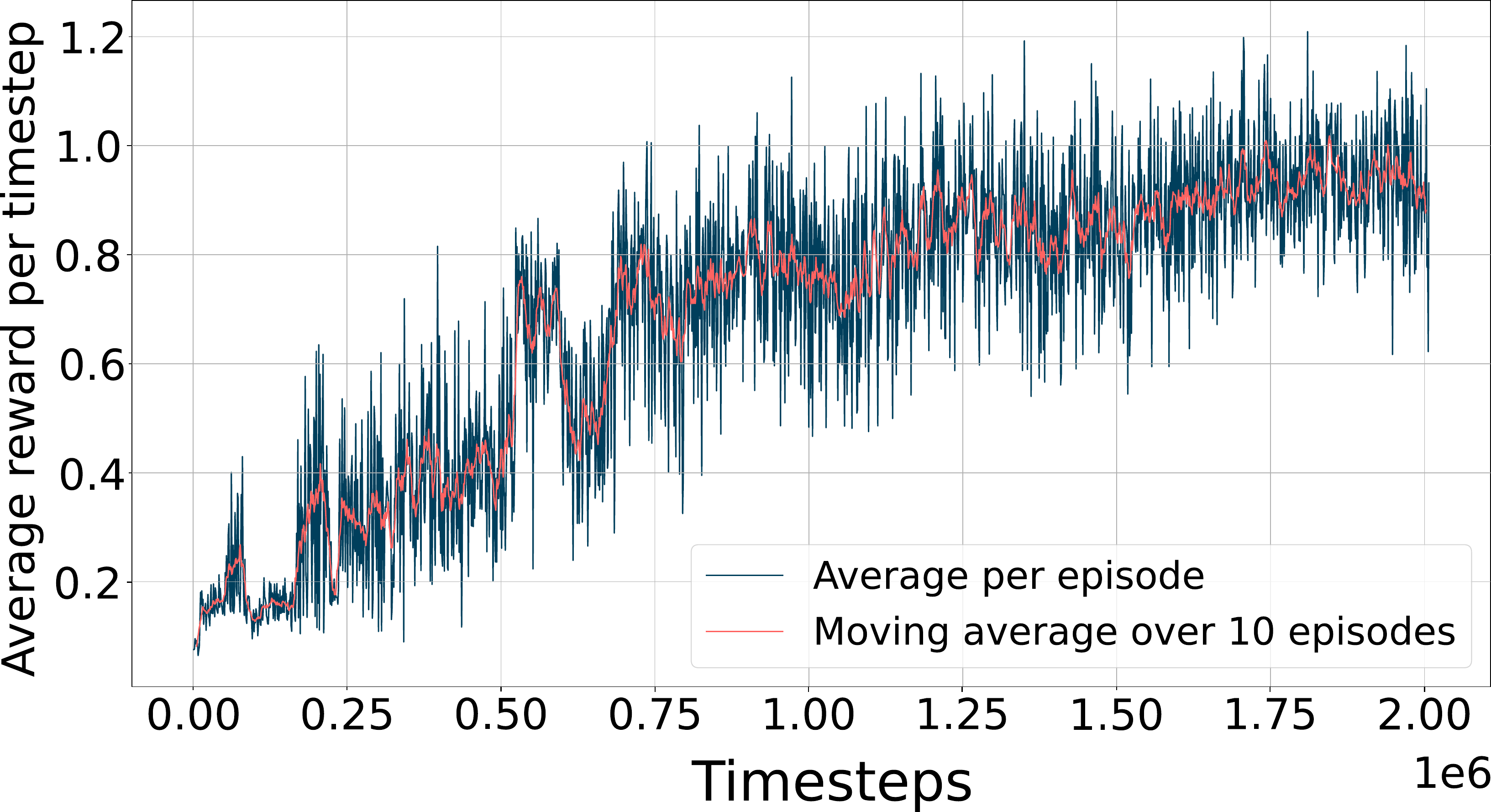}
    \captionsetup{belowskip=-0pt}\caption{Average reward per episode for the empowered agent ($c=1$). (Evaluated every $10^3$ training steps.)}
  \label{fig:avg-reward}
\end{figure}

\begin{table}[htb]
    \centering
    \renewcommand{\arraystretch}{1.2}
    \resizebox{\columnwidth}{!}{%
        \begin{tabular}{l | c | c}
            \toprule
             \multicolumn{1}{c}{\textbf{Agent}} & \multicolumn{1}{c}{\textbf{Success rate}} & \multicolumn{1}{c}{\textbf{Avg. clear time (s)}} \\
            \midrule
            \textbf{unempowered}      & 82\% & 6.56 \\
            \textbf{Empowered $c=10$}   & 89\% & 6.91 \\
            \textbf{Empowered $c=1$}    & 92\% & 7.77 \\
            \textbf{Empowered $c=0.5$}  & 99\% & 7.63 \\
            \bottomrule
        \end{tabular}%
    }
    \captionsetup{belowskip=-10pt}\caption{Success rate and time to clear the obstacle.}
    \label{tab:success_time}
\end{table}

\begin{figure}[htb]
    \centering
    \includegraphics[width=0.8\columnwidth]{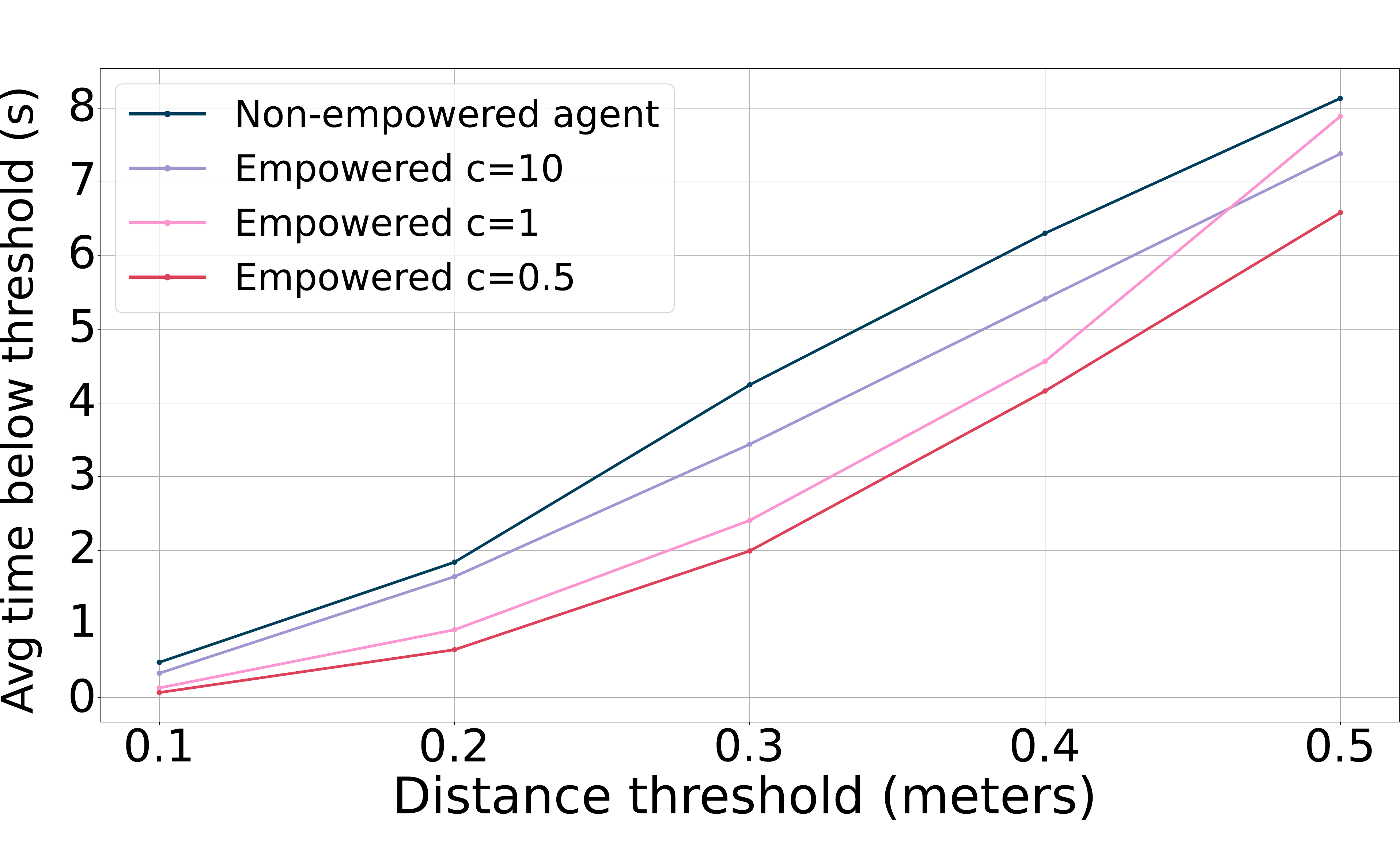}
    \captionsetup{belowskip=-0pt}\caption{Average time spent under different distance thresholds for empowered (blue lines) and unempowered (orange line) agents.}
  \label{fig:dist-threshold}
\end{figure}

\begin{table}[htb]
    \centering
    \renewcommand{\arraystretch}{1.2}
     \resizebox{\columnwidth}{!}{%
    \begin{tabular}{l | c c c c c}
        \toprule
        \textbf{Distance Threshold (m)} & \textbf{0.1} & \textbf{0.2} & \textbf{0.3} & \textbf{0.4} & \textbf{0.5} \\
        \midrule
        \textbf{Empowered $c=10$}  & 31.05\% & 10.70\% & 18.99\% & 14.14\% & 9.25\% \\
        \textbf{Empowered $c=1$}   & 72.61\% & 49.93\% & 43.34\% & 27.59\% & 3.01\% \\
        \textbf{Empowered $c=0.5$} & 85.76\% & 64.66\% & 53.10\% & 33.98\% & 19.07\% \\
        \bottomrule
    \end{tabular}}
    \captionsetup{belowskip=-10pt}\caption{Percentage decrease in time spent near obstacles for empowered agents relative to unempowered baseline.}
    \label{tab:time_decrease}
\end{table}

\section{Discussion}
\label{sec:discussion}
Established Safe RL methods \cite{gu2024review}, such as Control Barrier Functions (CBFs), typically ensure safety by defining a safe set in the state space and using a controller that guarantees the agent's state remains within this set, resulting in hard, explicit safety constraints. \emph{Again, we emphasize that our method can be used on top of established safe RL methods. In the future, we want to evaluate how this combination enhances various safe RL methods.}

Compared with standard safe RL methods that use additive reward penalties for unsafe behavior, C-STEP offers a philosophically distinct approach. Rather than imposing hard constraints to quantify safe states, it shapes the reward landscape to create an \textit{intrinsic motivation} for safety. The agent is not explicitly forbidden from entering certain states; instead, it is rewarded for maintaining a high degree of future maneuverability (i.e., a large collision-free reachable set). This encourages the agent to learn policies that \textit{emergently} exhibit safe behavior, such as maintaining distance from obstacles not because of a penalty function, but because doing so maximizes its future options. C-STEP can be viewed as a complementary tool that promotes a broader form of robustness, whereas CBFs provide formal safety guarantees with respect to a predefined safe set. A direct quantitative comparison against CBFs and other intrinsic rewards remains an avenue for future work to understand the trade-offs between these paradigms.

\subsection{On Model Free Learning}

Our method requires a simulator for evaluating finite-horizon trajectories. Notably, we do not require differentiable transition functions, as used by various safe RL methods and model-predictive controllers \cite{gu2024review}. In addition, our method does not require an explicit characterization of safe states; the simulator only needs to distinguish terminated and non-terminated trajectories, ands our method infers state safety from the volume sampled reachable set. In this way, our method may be considered to lie between pure model-free and model-based RL. 

\subsection{Beyond Deterministic Dynamics}

We formulated our method for a deterministic dynamical system. However, the implementation naturally extends to stochastic systems, since we only require a simulator to evaluate finite-horizon trajectories. Furthermore, most real-world robotic systems can be modeled as deterministic dynamical systems subject to environmental uncertainty and noise, and our theory should hold in approximation.

\section{Conclusion and Future Work}
\label{sec:concl}
This work introduced CST-Empowerment as an intrinsic reward for adapting navigation reward functions towards enhanced safety. We demonstrated that our physics-informed, empowered reward function efficiently enhances the safety of policies trained with RL for mobile robotic navigation.

We also acknowledge several limitations in our validation, which present directions for future research. First, a comprehensive evaluation would include direct quantitative comparisons against state-of-the-art baselines, statistical significance tests, and ablation studies. Furthermore, validating our approach on physical hardware would be essential to assess its robustness in dynamic, real-world environments.

Theoretical guarantees are also a natural next step. Other attractive future work includes training policies solely to maximize empowerment for use in critical situations and testing C-STEP in combination with other advanced reward designs. We believe empowerment is a powerful complementary tool that can augment any navigation reward to promote physics-informed intrinsic safety.

\section*{\uppercase{Acknowledgements}}

This work has received funding from the EU Horizon project SHEREC (HORIZON-CL4-2023-HUMAN-01-02).

\bibliographystyle{apalike}
{\small
\bibliography{references}}

\section*{\uppercase{Appendix}}

\subsection{Understanding Empowerment of Deterministic Agents}

To support our new definition of empowerment, consider a real-valued random variable $S' = g(S,A)$, which is a (measurable) function $g$ of real-valued random variables $S$ and $A$. Further, let $W_\varepsilon$  be a zero mean normal random variable with variance $\varepsilon^2>0$ independent of $(S,A)$. Then, define the perturbation $\Tilde{S}' \coloneqq g(S,A) + W_\varepsilon$. The conditional density of $\Tilde{S}'$ given $(S,A)$ is then given by $\cN(\Tilde{s}'; g(s,a), \varepsilon^2)  $. It follows that 
$h(\Tilde{S}' \mid S, A) = \ln \left( \varepsilon {\sqrt {2\,\pi \,e}}\right)$, the differential entropy of a normal random variable with arbitrary mean and variance $\varepsilon^2$. Therefore, the channel capacity is 
\begin{equation}
\label{eq:example_perturbation}
    \cC(A \to \Tilde{S}' \mid S) = \max\limits_{p(a)} h(g(S,A) \mid S) -   \ln \left( \varepsilon {\sqrt {2\,\pi \,e}}\right).
\end{equation}
Suppose that we use \eqref{eq:example_perturbation} %\cC(A \to \Tilde{S}' \mid S=s)$ as a function of $s$
to determine realizations of $S$ that maximize capacity. As $\ln \left( \varepsilon {\sqrt {2\,\pi \,e}}\right)$ is state independent follows that the maximizer of \eqref{eq:example_perturbation} is %equal to $\arg\max\limits_s \{ \max\limits_{p(a)} h(g(S,A) \mid S=s) \} $ for any $\varepsilon > 0$. 
independent of the perturbation, and the states that maximize empowerment are those that maximize the first term in \eqref{eq:example_perturbation}. Consequently, defining the empowerment of deterministic agents as the maximum state-dependent differential entropy is pertinent.

\subsection{Proof of \Cref{prop:1}}

By construction, the space of $X(T)$ given $x(0)=x_0$ is $\cR_T\left(x\right)$. As $f$ is Lipschitz continuous, the Picard–Lindelöf Theorem implies that every solution of Lipschitz continuous ODEs is unique. Hence, there exists a (non-trivial) distribution on $\cC([0,T];\cU)$, such that $X(T)$ is uniformly distributed on the reachable set. As the  differential entropy is maximized for uniform distributions, the statement follows.

\subsection{Simulation Details}

The inner area of the maps used in our evaluation set had dimensions of $5m \times 2.5m \times 2m$. The LiDAR sensor on the drone had $180$ evenly spaced beams covering $360^{\circ}$. The drone's start and goal z-axis positions were fixed at $1m$.
\paragraph{CTS-Empowerment Approximation}
 To model the drone's dynamics, we used its URDF file to incorporate relevant physical properties into our trajectory estimation.
\vspace{-5pt}
\begin{table}[htb]
    \centering    
    \begin{tabular}{l | l}
        \toprule
        Hyperparameters & Values \\
        \midrule
        $T$ & 1 second \\
        Trajectories  & 150 \\
        Points in trajectory & 100 \\        
        \bottomrule
    \end{tabular}
    \captionsetup{belowskip=-10pt}\caption{CST Empowerment hyperparameters used in experiments.}
   
     \label{tab:cst-hyper}
\end{table}

\paragraph{PPO Hyperparameters}
The network architecture consisted of two fully connected hidden layers of 2048 and 512 neurons, respectively.
\vspace{-5pt}
\begin{table}[htb]
    \centering   
    \begin{tabular}{l | l}
        \toprule
        Hyperparameters & Values \\
        \midrule
        Horizon & 2048 \\
        Batch size   & 64 \\
        Discount $(\gamma)$ & 0.99 \\
        GAE Parameter $(\lambda)$ & 0.95 \\
        Clip range & 0.2 \\
        Learning rate & $3 \times 10^{-5}$ \\
        SDE & True \\
        Num. epochs & 10 \\        
        \bottomrule
    \end{tabular}
   \captionsetup{belowskip=-0pt}\caption{PPO parameters for the PyBullet environment.}
    \label{tab:dronePPOhyper}
\end{table}

\end{document}